%% file: manuscript.tex
\documentclass[letterpaper, 10 pt, conference]{ieeeconf}  

\IEEEoverridecommandlockouts                              

\usepackage{graphicx} 
\graphicspath{{Figures/}}

\usepackage{amsmath}
\usepackage{amssymb}
\usepackage{helvet}
\usepackage[acronym, toc]{glossaries}
\usepackage{tikz,bm,color}
\usepackage{enumerate}
\usepackage{mathtools, cuted}
\usepackage{pgfplots}
\usepackage{pgfplotstable}
\usetikzlibrary{external,calc,patterns,decorations.pathmorphing,decorations.markings,decorations.text,arrows,shapes,positioning, backgrounds}
\usetikzlibrary{spy}
\usepgfplotslibrary{fillbetween}

\usepackage{hhline}
\usepackage{ctable}
\usepackage{pbox}

\definecolor{cadmiumgreen}{rgb}{0.0, 0.42, 0.24}
\definecolor{PIcolor}{rgb}{0, 0, 1.0}
\definecolor{STSMCcolor}{rgb}{0.929, 0.427, 0.067}
\definecolor{NACcolor}{rgb}{0.0, 0.42, 0.24}
\definecolor{InIcolor}{rgb}{0.247, 0.475, 0.749}
\definecolor{limitColor}{rgb}{0.247, 0, 0}
\pgfplotsset{compat=newest}
\usepackage{balance}

\newtheorem{proposition}{Proposition}

\usepackage{algorithm}
\usepackage{algpseudocode}
\usepackage[normalem]{ulem}
\usepackage{slashbox}
\newcommand\T{\rule{0pt}{2.6ex}}       
\newcommand\B{\rule[-1.2ex]{0pt}{0pt}} 

\input{abbreviations}

\title{\LARGE \bf
Towards Prescribed Accuracy in Under-tuned Super-Twisting Sliding Mode Control Loops - Experimental Verification
}

\author{Dimitrios Papageorgiou$^{1}$
\thanks{$^{1}$Dimitrios Papageorgiou is with the Department of Electrical and Photonics Engineering, Technical University of Denmark, Elektrovej 326, 2800 Kgs Lyngby, Denmark
        {\tt\small dimpa@elektro.dtu.dk}}%
}

\begin{document}
\setlength\tabcolsep{3 pt}

\maketitle
\thispagestyle{empty}
\pagestyle{empty}

\begin{abstract}
	Obtaining prescribed accuracy bounds in super-twisting sliding mode control loops often falls short in terms of the applicability of the controller in high-performance systems. This is due to the fact that the selection of the controller gains that are derived from the conditions for finite-time convergence may be too restrictive in connection to actuator limitations and induced chatter. Previous work has shown that in case of periodic perturbations, there can be a systematic selection of much lower controller gains that guarantees boundedness of the closed-loop solutions within predetermined accuracy bounds. This study presents an experimental validation of these findings carried out on a commercial industrial motor system.
\end{abstract}

\section{INTRODUCTION}
	\input{introduction}

\section{OVERVIEW OF THEORETICAL RESULTS} \label{sec:theory}
	\subsection{Preliminaries}
	\input{systemDescritpion}
	
	\subsection{Existence and properties of limit cycles}
	\input{stabilityAnalysis}
	
	\subsection{Tuning based on accuracy specifications}
	\input{tuningGuidlines}

\section{EXPERIMENTAL VERIFICATION} \label{sec:experimental}
	This section presents the physical system and scenarios used for validating the theoretical findings of the previous sections and discusses the obtained results. The experimental setup comprised a commercial Siemens 1FT7042-5AF70 \gls{PMSM} equipped with a SINAMICS S120 drive converter with 11-bit IC22DQ incremental angle encoder.
	The test scenarios covered two regimes of motion: constant speed and sinusoidal velocity profile. The total perturbation $d(t)$ acting on the motor is the sum of friction and cogging torques given by \cite{papageorgiou2020online,lukaniszyn2004optimization}:
	\begin{align}
		d(t) &= T_C \frac{2}{\pi}\arctan(\alpha\omega) + \beta\omega + \sum\limits_{i = 1}^{N}F_i\sin(\theta + \psi_i)\\
		\dot{d}(t) &= \left[ \frac{2T_C \alpha}{\pi\left( 1 + \alpha^2\omega^2 \right)} + \beta \right]\dot{\omega} + \omega\sum\limits_{i = 1}^{N}F_i\cos(\theta + \psi_i) \label{eq:disturbance_derivative}
	\end{align}
	where $\omega,\theta$ are the angular velocity and position of the motor, $T_C,\beta$ are the Coulomb and viscous friction coefficients, $\alpha >> 1$ is the steepness factor of the Coulomb friction model for approximating the signum function and $F_i, \psi_i$ are the amplitude and phase offset of the cogging torque component related to the $i^{\text{th}}$ harmonic. In all the experiments the signal of interest was the error $e \triangleq \omega - \omega_r$ and $d,\dot{d}$ were estimated via robust differentiation of the measured velocity and control input $u$ as $d = J\dot{\omega} - u$, where $J$ is the motor inertia.

	\subsection{Speed regulation at constant set-point}
	\input{constantSpeedRegime}
	
	\subsection{Tracking of sinusoidal velocity reference}
	\input{sinusoidalSpeedRegime}

\section{CONCLUSIONS} \label{sec:conclusions}
	\input{conclusions}




%

\bibliographystyle{IEEEtran}        
\bibliography{IEEEabrv,Bibliography/mybibl}

\end{document}

%% file: abbreviations.tex
\newacronym{FTC}{FTC}{Fault-Tolerant Control}
\newacronym{AC}{AC}{Alternating Current}
\newacronym{PMSM}{PMSM}{Permanent Magnet Synchronous Motor}
\newacronym{PMSMs}{PMSMs}{Permanent Magnet Synchronous Motors}
\newacronym{CNC}{CNC}{Computer Numerical Control}
\newacronym{SMO}{SMO}{Sliding Mode Observer}
\newacronym{SMC}{SMC}{Sliding Mode Control}
\newacronym{HOSMC}{HOSMC}{High-order Sliding Mode Control}
\newacronym{ISS}{ISS}{Input-to-State Stable}
\newacronym{LMI}{LMI}{Linear Matrix Inequality}
\newacronym{EKF}{EKF}{Extended Kalman Filter}
\newacronym{SISO}{SISO}{Single-Input Single-Output}
\newacronym{PID}{PID}{Proportional-Integral-Differential}
\newacronym{PI}{PI}{Proportional-Integral}
\newacronym{P}{P}{Proportional}
\newacronym{rpm}{rpm}{rounds pre minute}
\newacronym{STSMO}{STSMO}{Super-twisting Sliding Mode observer}
\newacronym{STSMC}{STS\-MC}{Super-twisting Sliding Mode controller}
\newacronym{PE}{PE}{Persistence of Excitation}
\newacronym{AO}{AO}{Adaptive Observer}
\newacronym{UGAS}{UGAS}{Uniformly Globally Asymptotically Stable}
\newacronym{BIBO}{BIBO}{Bounded Input-Bounded Output}
\newacronym{ULES}{ULES}{Uniformly Locally Exponentially Stable}
\newacronym{UB}{UB}{Uniformly Bounded}
\newacronym{UGB}{UGB}{Uniformly Globally Bounded}
\newacronym{MAEE}{MAEE}{Maximum Absolute Estimation Error}
\newacronym{VSC}{VSC}{Variable-Structure Control}
\newacronym{DC}{DC}{Direct Current}

\makenoidxglossaries

%% file: introduction.tex
\gls{VSC} and in particular \gls{SMC} techniques \cite{emelyanov1967variable} constitute a quite attractive family of control strategies due to their inherent robustness against bounded and bounded-rate perturbations. The additional feature of ensuring finite-time convergence has established the use of such algorithms in many fields ranging from aviation and flight control \cite{shtessel2002tailless} to industrial motion control systems \cite{utkin2017sliding}, while their application also extends to diagnosis and fault-tolerant control schemes \cite{Edwards2000}. The need to alleviate the induced chatter in the control signal of the conventional first-order \gls{SMC} led to the development of second and higher-order \gls{SMC} laws \cite{bartolini661074,levant2005homogeneity}. Among these, the \gls{STSMC} introduced in \cite{levant1993sliding} and further generalised in \cite{haimovich2019a}, has become very popular due to its robust finite-time stabilisation properties and reduced chattering \cite{Levant2007576}.

The simplicity of the \gls{STSMC} with respect to its design and implementation has motivated a large number of studies on the systematic tuning of the controller. Relating the selection of controller gains to performance specifications is of particular interest since it facilitates easy commissioning of control systems. Despite the ``proportional-integral" structure of the controller, the tuning of the \gls{STSMC} can be challenging and has therefore received significant attention. Closed-form expressions for the controller gains were provided by \cite{Moreno2012} who used strict Lyapunov functions to prove finite-time stability of the \gls{STSMC} closed loop. Necessary and sufficient conditions for finite-time convergence were provided in \cite{behera2018new} and \cite{seeber2018necessary}, in which the authors employed geometric arguments relating to the majorant curve contraction requirement. The same conditions were also used to estimate the reaching time as shown in \cite{seeber2018a}. Tuning rules based on the properties of limit cycles appearing in linear systems with uncertain actuator dynamics were provided by \cite{pilloni2012a}, who used a describing functions framework. Adaptation constitutes an alternative approach to systematic commissioning of \gls{STSMC} loops. The authors in \cite{shtessel2012a} proposed an adaptive \gls{STSMC} design for an electropneumatic actuator. A certainty-equivalence adaptive \gls{STSMC} was presented in \cite{barth2015a}, where the adaptation was used to avoid unnecessary large controller gains. The authors in \cite{edwards2016adaptive} and \cite{edwards2016a} proposed a dual-layer adaptive \gls{STSMC} for guaranteeing finite-time convergence to the origin for both known and unknown perturbation bounds.

In the majority of the foregoing studies, the selection of the \gls{STSMC} gains was based on the conditions for finite-time convergence of the error variable to the origin, which require that the integral gain of the \gls{STSMC} be larger than the bound of the rate of the lumped perturbations affecting the system dynamics \cite{Moreno2012}. This can be limiting in terms of unrealistically large control signals, especially for application in systems with abrupt-changing perturbations such as Coulomb friction and backlash torques during motion reversals in mechanical systems. However, ensuring finite-time convergence is not necessary for obtaining high accuracy as was demonstrated in \cite{papageorgiou2019adaptive}, where an under-tuned \gls{STSMC} outperformed several conventional and advanced controllers in a single-axis positioning task. For systems affected by periodic perturbations it was proven \cite{papageorgiou2021behaviour} that under milder gains conditions compared to those of finite-time convergence, the solutions of the \gls{STSMC} closed-loop system converge to a limit cycle of the same period as the perturbation. The width of the limit cycle, which is a bound for the control error variable, can be modulated according to prescribed accuracy requirements. The authors provided guidelines for systematic tuning of the \gls{STSMC} and demonstrated the validity of the method in simulation.

This study pursues experimental verification of the theoretical results presented in \cite{papageorgiou2021behaviour}. Specifically, a motion control system comprising a commercial \gls{PMSM} is used as a test platform in a series of experiments proving the existence of stable limit cycles in under-tuned \gls{STSMC} loops with perturbation-depended frequency and amplitude characteristics. Moreover, the performance of the closed-loop system is assessed in connection to the tuning method. The remainder of the paper is organised as follows: Section \ref{sec:theory} provides an overview of the theoretical findings in \cite{papageorgiou2021behaviour}. Section \ref{sec:experimental} details the experimental campaign and discusses the results in connection to the theoretical predictions. Finally, conclusions are drawn in Section \ref{sec:conclusions} along with some remarks on future work.

%% file: systemDescritpion.tex
The study concerns the class of nonlinear \gls{SISO} systems described by
\begin{equation}
	\dot{y} = h(t,y) + g(t,y)u_0 + d(t)
\end{equation}
where $y\in\mathbb{R}$ is available from measurements, the scalar functions $h(t,y),g(t,y)\in\mathcal{C}^1$ are bounded for bounded $y$, $g(t,y) \neq 0, \; \forall (t,y)\in[0,\infty)\times\mathbb{R}$ and $d(t)\in\mathcal{C}^2$ is a $T$-periodic unknown function. 
Such systems are frequently encountered in industrial applications that include repeated closed-curve tracking such as machine tool drive axes \cite{altintas2001manufacturing}, where $d(t)$ could be the effect of Coulomb friction and cogging torques on the drive motor and axis dynamics that cause contouring deformations \cite{gross2001electrical}. The control law
\begin{align}
	u_0 &= g^{-1}(t,y)\left[ -h(t,y) + u \right]\\
	u &= -k_1\vert y \vert^{\frac{1}{2}}\text{sgn}(y) - k_2\int_0^t \text{sgn}(y(\tau))d\tau \label{eq:control_law} \; ,
\end{align}
where \textnormal{sgn}$(\cdot)$ is the signum function, gives the following second-order closed-loop dynamics
\begin{equation}\label{eq:closed_loop_system}
	\underbrace{\begin{bmatrix}
		\dot{x}_1\\
		\dot{x}_2
		\end{bmatrix}}_{\boldsymbol{\dot{x}}} = \underbrace{\begin{bmatrix}
			-k_1\vert x_1 \vert^{\frac{1}{2}}\text{sgn}(x_1) + x_2\\
			-k_2\text{sgn}(x_1) + q(t)
		\end{bmatrix}}_{\boldsymbol{f}(t,\boldsymbol{x})} \; ,
\end{equation} 
with $x_1 \triangleq y$, $x_2 \triangleq -k_2\int_0^t \text{sgn}(y(\tau))d\tau + d(t)$ and $q(t) \triangleq \dot{d}(t)$. Since $d(t)\in\mathcal{C}^2$ and is $T$-periodic, it follows that its derivative is a continuous bounded $T$-periodic function. Let $\vert q(t) \vert \leq L$, where $L > 0$. It has been shown \cite{seeber2018necessary} that if
\begin{align}
	k_2 &> L \label{eq:finite_time_conditions_k_2}\\
	k_1 &\geq 1.8\sqrt{k_2 + L}, \label{eq:finite_time_conditions_k_1}
\end{align}
then the system in \eqref{eq:closed_loop_system} has a unique finite-time stable equilibrium point at the origin. The use of $\text{sgn}(y)$ implies infinitely fast switching of the control signal, which is not feasible in real-life control systems due to actuator limitations. In practice, the signum function is approximated by a continuous ``boundary layer" function such as the following \cite{llibre2015birth}
\begin{equation} \label{eq:phi_definition}
	\phi_{\delta}(q,\delta) \triangleq \begin{cases}
	1 &\text{ if } q \geq \delta\\
	\displaystyle\frac{q}{\delta} &\text{ if } -\delta < q < \delta\\
	-1 &\text{ if } q \leq -\delta
	\end{cases} \; ,
\end{equation}
where $\delta$ is the width of the boundary layer. By doing so, the discontinuous vector field $\boldsymbol{f}(t,\boldsymbol{x})$ in \eqref{eq:closed_loop_system} is \emph{regularised}, i.e. is approximated by a continuous vector field, which often simplifies the analysis of the closed-loop system. It was proven in \cite{papageorgiou2021behaviour} that such an approximation can be made with arbitrarily large accuracy by letting $\delta \rightarrow 0$. The practical implication of this approximation is that results relating to convergence to the origin now correspond to convergence to a neighbourhood of the origin or arbitrarily small size (depending on $\delta$). This is also what actually happens in real systems due to the effect of noise and other model inaccuracies. The regularised vector field will be considered in the entire subsequent analysis.

%% file: stabilityAnalysis.tex
Consider the regularised system $\boldsymbol{\dot{x}} = \boldsymbol{f}_{\delta}(t,\boldsymbol{x})$ with
\begin{align}\label{eq:regularised_vector_field}
	\boldsymbol{f}_{\delta}(t,\boldsymbol{x}) &\triangleq \begin{bmatrix}
		-k_1\vert x_1 \vert^{\frac{1}{2}}\phi_{\delta}(x_1,\delta) + x_2\\
		-k_2\phi_{\delta}(x_1,\delta) + q(t)
	\end{bmatrix} \; ,
\end{align}
where $k_1, k_2 > 0$, $k_2 < L$ and $\phi_{\delta}:\mathbb{R}\times (0,+\infty) \rightarrow [-1,1]$ defined in \eqref{eq:phi_definition}.
The following proposition states the conditions under which the boundedness of the regularised (and by extension of the real) system is ensured by means of convergence to a limit cycle even though the finite-time stability conditions do not hold ($k_2 < L$). Obviously, the size of the limit cycle along $x_1 = 0$, i.e. the bound on $x_1(t)$, is a metric for the closed-loop system accuracy.

\begin{proposition}[\cite{papageorgiou2021behaviour}]
	Consider the closed-loop system \eqref{eq:closed_loop_system} and its approximation associated with the regularisation \eqref{eq:regularised_vector_field}, where $q(t)$ is Lipschitz, $T$-periodic of sufficiently small period $T$ and $\vert q(t) \vert \leq L$. Then, $\exists \varepsilon_1 > 0$ with  $0 < T <\varepsilon_1$ such that under the conditions 
	\begin{align}
		k_2 &> \left\vert \frac{1}{T}\int_{0}^{T}q(t)dt \right\vert \label{eq:convergence_condition_k2}\\
		k_1 &\geq 1.8\sqrt{k_2 + \left\vert \frac{1}{T}\int_{0}^{T}q(t)dt \right\vert} \label{eq:convergence_condition_k1}
	\end{align}
	the trajectories of the regularised system $\boldsymbol{\dot{x}} = \boldsymbol{f}_{\delta}(t,\boldsymbol{x})$ converge to a limit cycle with period $T$.
\end{proposition}
\begin{proof}
	The regularised system can be written as
	\begin{equation}\label{eq:averaging_form}
		\boldsymbol{\dot{x}} = \varepsilon \frac{1}{T}\boldsymbol{f}_{\delta}(t,\boldsymbol{x}) \triangleq \varepsilon \boldsymbol{g}(t,\boldsymbol{x}) \; , \; \varepsilon = T \; ,
	\end{equation}
	where $\boldsymbol{g}(t,\boldsymbol{x})$ is Lipschitz continuous. The associated averaged system is written as 
	\begin{equation}
		\boldsymbol{\dot{\chi}} = \varepsilon \boldsymbol{\bar{g}}(\boldsymbol{\chi}), \; \boldsymbol{\chi} = \begin{bmatrix}
			\chi_1 & \chi_2
		\end{bmatrix}^T\in \mathbb{R}^2
	\end{equation}
	with $\varepsilon \boldsymbol{\bar{g}}(\boldsymbol{\chi}) = \displaystyle\frac{1}{T}\int_{0}^{T}\boldsymbol{f}_{\delta}(t,\boldsymbol{x}) dt$ and finally
	\begin{equation}\label{eq:averaged_system}
		\boldsymbol{\dot{\chi}} = \begin{bmatrix}
			-k_1\vert \chi_1 \vert^{\frac{1}{2}}\phi_{\delta}(\chi_1,\delta) + \chi_2\\
			-k_2\phi_{\delta}(\chi_1,\delta) + \displaystyle\frac{1}{T}\int_{0}^{T}q(t) dt
		\end{bmatrix} \; .
	\end{equation}
	Comparing \eqref{eq:averaged_system} to \eqref{eq:regularised_vector_field} reveals that if conditions \eqref{eq:convergence_condition_k2} and \eqref{eq:convergence_condition_k1} hold, then for sufficiently small $\delta$ ($\delta \rightarrow 0$) the origin is a finite-time stable equilibrium point of the averaged system. Then, by Theorem 4.1.1 in \cite{guckenheimer1983a}, there exists $\varepsilon_1 > 0$, such that $\forall \varepsilon\in(0,\varepsilon_1)$, the solutions of \eqref{eq:averaging_form} converge to a unique isolated $T$-periodic orbit $\gamma_{\varepsilon}(t) = O(\varepsilon)$ of same stability type.
\end{proof}
The constant $\varepsilon_1$ relates to the largest time scale $\frac{1}{\varepsilon_1}$ for which approximation of the system dynamics via averaging is practically valid. Conditions \eqref{eq:convergence_condition_k2} and \eqref{eq:convergence_condition_k1} are much less strict than those for finite-time convergence since in case of symmetric perturbation rates ($\int_{0}^{T}q(t)dt = 0$), boundedness of the solutions is ensured by merely selecting positive gains irrespectively of how large the perturbation rate may be.

The remaining analysis concerns the size of the limit cycle, i.e. the bound on $x_1(t)$, and its relation to the perturbation characteristics and controller gains. To proceed with this analysis, it is convenient to express the closed-loop dynamics in the phase space coordinates $w_1 \triangleq x_1$ and $w_2 \triangleq \dot{x}_1$ as
\begin{align}
	\dot{w}_1 &=  w_2 \label{eq:w_1_dot}\\
	\dot{w}_2 &=  -\frac{1}{2}k_1\vert w_1 \vert^{-\frac{1}{2}}w_2 - k_2\text{sgn}(w_1) + q(t) \label{eq:w_2_dot} \; .
\end{align}

\begin{proposition}[\cite{papageorgiou2021behaviour}]\label{prop:bound_period}
	After the trajectories of the closed-loop system converge to the limit cycle, the bound on the state $x_1$ varies proportionally to the perturbation bound $L$ and to the square of the perturbation period $T$.
\end{proposition}
\begin{figure}[t]
	\begin{center}
		\includegraphics[width = 0.425\textwidth]{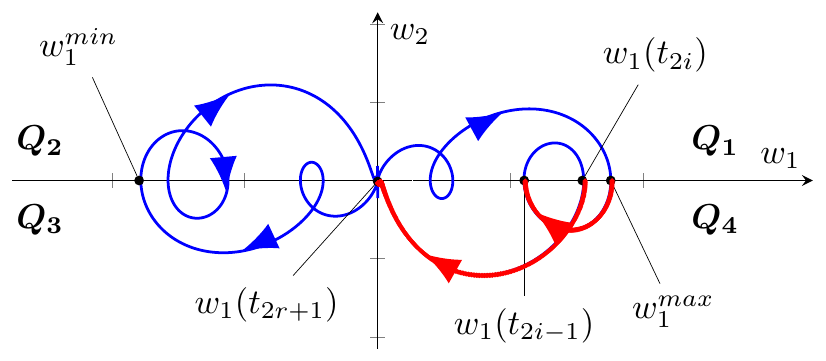}
		\vspace{-9pt}
		\caption{One full period of the limit cycle.}\label{fig:onePeriod}
	\end{center}
\vspace{-10pt}
\end{figure}
\begin{proof}
	Le $Q_j, \; j = 1,\dotsc, 4$ be the four quadrants of the phase space and consider one period of the limit cycle as shown in Figure \ref{fig:onePeriod}. For $t \geq t_0$ assume that the trajectories $\boldsymbol{w}(t)$ intersect with semi-axis $w_1 \geq 0$ at $2r + 1$ points, $r\in\mathbb{N}$ starting from $w_1^{max} \triangleq w_1(t_0)$, which is the maximum value of $w_1(t)$. Since the trajectories cannot cross from $Q_1$ to $Q_2$ (due to increasing $w_1$), each trajectory segment that lies in $Q_4$ starting at an intersection point will have to either cross the semi-axis $w_1 \geq 0$ twice (one while crossing to $Q_1$ and one right after while crossing to $Q_4$) or cross the vertical axis towards $Q_3$. In both cases, there will always be an odd number of intersections with the semi-axis $w_1 \geq 0$. Let these intersections occur at time instances $t_{2i}$ (from $Q_1$ to $Q_4$) and $t_{2i+1}$ (from $Q_4$ to $Q_1$), $i\in\mathcal{I} \triangleq \{0,\dotsc,r\}$ with $\boldsymbol{w}(t)$ crossing from $Q_4$ to $Q_3$ at $t = t_{2r+1}$. In each time interval $[t_{2i},t_{2i+1}]$, where $w_1(t) > 0, w_2(t) \leq 0$ holds (red lines):
	\begin{align}
		\dot{w}_2(t) &\geq -(k_2 + L) \Rightarrow w_2(t) \geq -(k_2 + L)(t - t_{2i}), \label{eq:w_2_t_3_bound}
	\end{align}
	$\forall t\in(t_{2i},t_{2i+1}]$ since $w_2(t_{2i}) = 0$. This leads to
	\begin{align*}
		\int_{t_{2i}}^{t_{2i+1}}w_2(t)dt &> -\frac{1}{2}(k_2 + L)(t_{2i+1} - t_{2i})^2, \; i\in\mathcal{I}.
	\end{align*}
	Since $w_1(t_{2r+1}) = 0$ and $w_2(t) \geq 0, \; \forall t\in[t_{2i-1},t_{2i}], \; i\in\mathcal{I}-\{0\}$, it follows that
	\begin{align*}
		&w_1^{max} = -\left( \sum\limits_{i=0}^{r}\int_{t_{2i}}^{t_{2i+1}}w_2(t)dt + \sum\limits_{i=1}^{r}\int_{t_{2i-1}}^{t_{2i}}w_2(t)dt \right)\\
		&\leq \frac{1}{2}(k_2 + L)\sum\limits_{i=0}^{r}(t_{2i+1} - t_{2i})^2\\
		&< \frac{1}{2}(k_2 + L)\left( \sum\limits_{i=0}^{r}(t_{2i+1} - t_{2i}) \right)^2 < \frac{1}{2}(k_2 + L)n^2T^2
	\end{align*}
	given that $t_{2i+1} > t_{2i}, \; \forall i\in\mathcal{I}$ and $\sum\limits_{i=0}^{r}(t_{2i+1} - t_{2i}) < t_{2r+1} - t_0 \leq nT, \; 0 < n \leq \frac{1}{2}$. Due to the homogeneity of the \gls{STSMC} closed-loop system \cite{seeber2018necessary}, the same analysis in $Q_2, Q_3$ gives a similar result for the minimum value that $x_1$ assumes, which finally leads to
	\begin{equation}\label{eq:w_1_bound_1}
		\max\limits_{x(t)\in\gamma_{\varepsilon}(t)} \vert x_1(t) \vert < \frac{1}{2}(k_2 + L)n^2T^2 \; , \; 0 < n \leq \frac{1}{2} \; .
	\end{equation}
\end{proof}
The result in inequality \eqref{eq:w_1_bound_1} implies that faster perturbations have less effect on the accuracy bounds since they are ``better averaged", whereas, as expected, larger perturbations compromise accuracy. Moreover, from \eqref{eq:w_2_t_3_bound} it can be seen that too large value for the integral gain $k_2$ will result in faster changes in $x_1$, i.e. in an increase of the induced chatter.

%% file: tuningGuidlines.tex
The description of the closed-loop dynamics in phase coordinates introduced in the previous section allows for a straightforward expression of the bound on $w_1(t) = x_1(t)$ as a function of the controller gains and, conversely, for an systematic tuning of the \gls{STSMC} given a accuracy specification.

Consider again a full period of the limit cycle in Figure \ref{fig:onePeriod} restricted in $Q_1$ and the time interval $\mathcal{T} \triangleq [t_0,t_m]$, where $t_0$ is the time when the trajectories first enter $Q_1$ from $Q_2$, $w_1(t_m) = w_1^{max} \geq w_1(t), \; \forall t\in\mathcal{T}$ and the time instant $t^*\in\mathcal{T}$ such that $w_2(t^*) \geq w_2*(t), \; \forall t\in\mathcal{T}$. Integrating Equation \eqref{eq:w_2_dot} over the interval $[t_0,t_m]$, where $w_1(0) = 0, \; w_2(0) > 0$, $w_1(t_m) = w_1^{max}, \; w_2(t_m) = 0$ and $\dot{w}_2(t^*) = 0$ leads to
\begin{align}
	&\int_{0}^{t_m}\dot{w}_2(t)dt = \int_{0}^{t_m}(q(t) - k_2)dt - \int_{0}^{t_m} k_1\frac{\dot{w}_1(t)}{2\sqrt{w_1(t)}}dt \nonumber\\
	&\Rightarrow -w_2(0) \leq -k_1\sqrt{w_1^{max}} + (L - k_2)t_m \Rightarrow \nonumber\\
	&\sqrt{w_1^{max}} \leq \frac{w_2(0) + (L - k_2)t_m}{k_1} \; . \label{eq:full_int}
\end{align}
Evaluating Equation \eqref{eq:w_2_dot} at $t = t^*$ gives
\begin{align} 
	w_2(t^*) &= 2\frac{q(t^*) - k_2}{k_1}\sqrt{w_1(t^*)} \; . \label{eq:max_w}	
\end{align}
Since $w_2(t^*) \geq w_2(t), \; \forall t\in [t_0,t_m]$, Equation \eqref{eq:max_w} yields
\begin{align*}
	w_2(0) &\leq w_2(t^*) \leq \frac{2(L - k_2)}{k_1}\sqrt{w_1^{max}}
\end{align*}
and from \eqref{eq:full_int} one obtains
\begin{align}
	&\sqrt{w_1^{max}} \leq \frac{\frac{2(L - k_2)}{k_1}\sqrt{w_1^{max}} + (L - k_2)t_m}{k_1} \Rightarrow \nonumber\\
	&\frac{k_1^2 - 2(L - k_2)}{k_1}\sqrt{w_1^{max}} \leq k_1(L - k_2)nT, \; 0 < n \leq \frac{1}{2} \; , \label{eq:ineq_2}
\end{align}
since $0 < t_m \leq \frac{T}{2}$. Finally, if $k_1$ is selected such that
\begin{equation}\label{eq:condition_k_1}
	k_1 > \sqrt{2(L - k_2)}
\end{equation}
then a not overly conservative bound for $w_1^{max}$ is given by
\begin{equation}\label{eq:w_1_bound}
	w_1^{max} \leq \frac{k_1^4(L - k_2)^2 n^2T^2}{\left[ k_1^2 - 2(L - k_2) \right]^2} \triangleq W_1(k_1,k_2) \; .
\end{equation}
Given an accuracy specification $\vert x_1(t) \vert \leq \eta$, $k_1$ and $k_2$ can be obtained by means of numerical optimisation or by selecting one gain and solving for the other. This is often the case in electromechanical systems, where $k_1$ represents forces and currents delivered by the actuator and it is desired to keep the level of actuation below some rated values. In such case, fixing $k_1$ allows for calculating $k_2$ from \eqref{eq:w_1_bound}:
\begin{equation} \label{eq:k_2_estimate}
	k_2 \geq L - \frac{\sqrt{\eta}k_1^2}{2\sqrt{\eta} + k_1 nT}
\end{equation}

%% file: constantSpeedRegime.tex
\begin{table}[t]
	\caption{Real and estimated error bound in constant speed regime.}\vspace{-10pt}
	\label{tab:constant_speed_scenarios}
	\begin{center}
		\input{Tables/constant_speed_scenarios}
	\end{center}
\end{table}
\begin{table}[t]
\caption{Real and estimated error bound for sinusoidal velocity.}\vspace{-10pt}
\label{tab:sinusoidal_speed_scenarios}
\begin{center}
	\input{Tables/sinusoidal_speed_scenarios}
\end{center}
\end{table}

\begin{figure}[t]
\begin{center}
	\includegraphics[width = 0.425\textwidth]{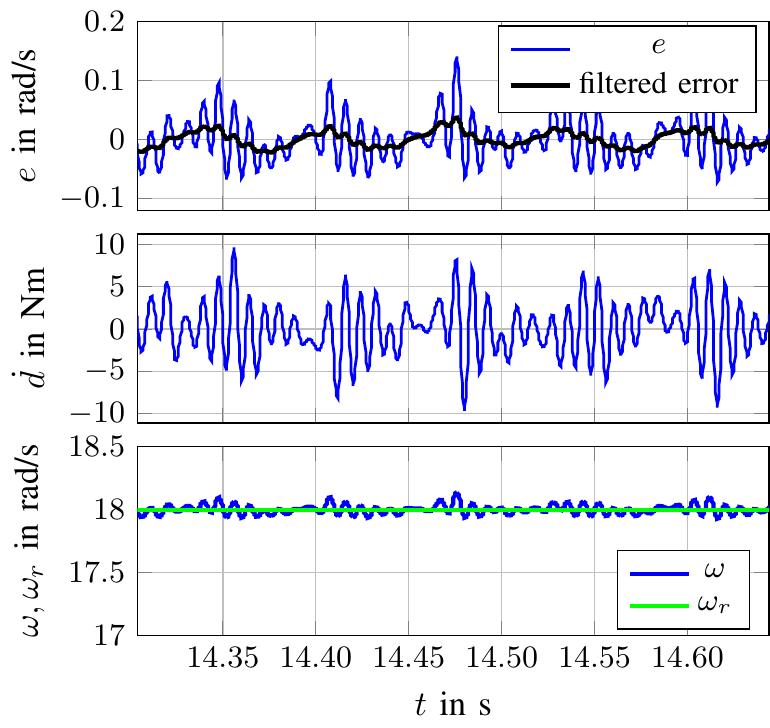}
	\caption{Error, disturbance rate and motor velocity for reference $\omega_r = 18$ rad/s. The periodicity of the error matches this of the perturbation rate.}\label{fig:signals_constant_speed}
\end{center}
\end{figure}

\begin{figure}[bp]
\begin{center}
	\includegraphics[width = 0.425\textwidth]{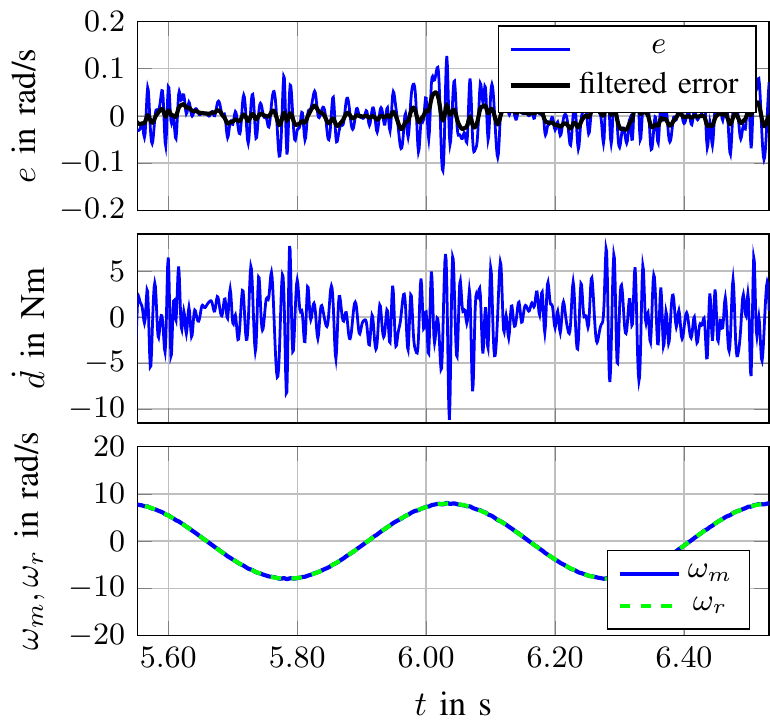}
	\caption{Error, disturbance rate and motor velocity for sinusoidal reference at 2 Hz. The periodicity of the error matches this of the perturbation rate.}\label{fig:signals_sinusoidal}
\end{center}
\end{figure}
While the motor operates at almost constant angular velocity $\omega_r$, then $\dot{\omega} \approxeq 0$ and from \eqref{eq:disturbance_derivative} one gets $\vert \dot{d}(t) \vert \leq \left\vert \omega_r \sum\limits_{i = 1}^{N}F_i \right\vert$ $\triangleq L$. Moreover, the period of $d(t)$ is approximately given by $T = 2\pi/\omega_r$. A total number of 12 experiments with constant speed set-points $\omega_r\in\{ 12, 13, \dotsc, 23 \}$ rad/s were carried out to assess the validity of the theoretical results. It was not possible to keep constant perturbation rate amplitude for all the experiments since both $\omega_r$ and the harmonics amplitudes $F_i$ were varying at different frequencies of operation. It was, however, observed that $L$ ranged from approximately 8 to 11 Nm/s and never exceeded 12 Nm/s, whereas its mean value over a single period was approximately 0. The finite-time convergence gains were calculated from \eqref{eq:finite_time_conditions_k_2},\eqref{eq:finite_time_conditions_k_1} as $\bar{k}_1 = 9.1$ and $\bar{k}_2 = 13.2$. The gains applied to the system were selected as $k_1 =0.9$ and $k_2 = 11.65 < L$ according to \eqref{eq:k_2_estimate} for $\eta = 0.2$ rad/s. Figure \ref{fig:limit_cycles_constant_speed} illustrates the phase plots for all the experiments with constant speed. As it can be seen, the closed-loop trajectories converge to a limit cycle. Figure \ref{fig:signals_constant_speed} that shows the time responses of $e,\dot{d}$ and $\omega$ during a full period for $\omega_r = 18$ rad/s, reveals that the signals have the same periodicity as predicted from the theoretical analysis. For increasing perturbation frequency, the actual error bound quadratically decreased as also shown in the bottom Figure \ref{fig:w_1_T}. Finally, Table \ref{tab:constant_speed_scenarios} shows a comparison between the real error bound and the one estimated based on \eqref{eq:w_1_bound_1} with $n = 0.5$.

\begin{figure}[bp]
	\begin{center}
		\includegraphics[width = 0.425\textwidth]{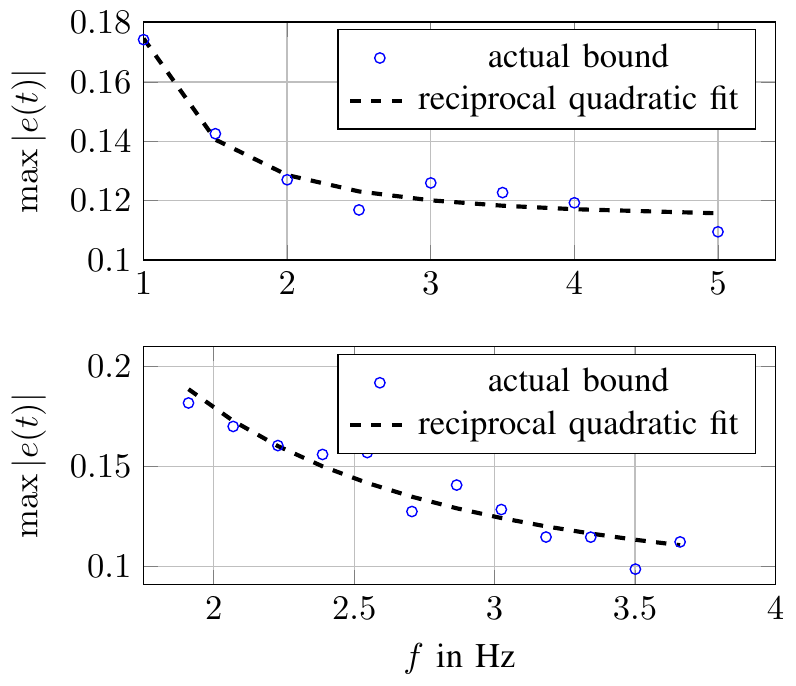}
		\caption{Velocity error bound as a function of the perturbation frequency $f$ for sinusoidal velocity reference (top) and constant velocity reference (bottom). In both cases, the error shows a quadratic decrease with respect to the frequency.}\label{fig:w_1_T}
	\end{center}
\end{figure}

\begin{figure*}[t]
	\begin{center}
		\includegraphics[width = 0.95\textwidth]{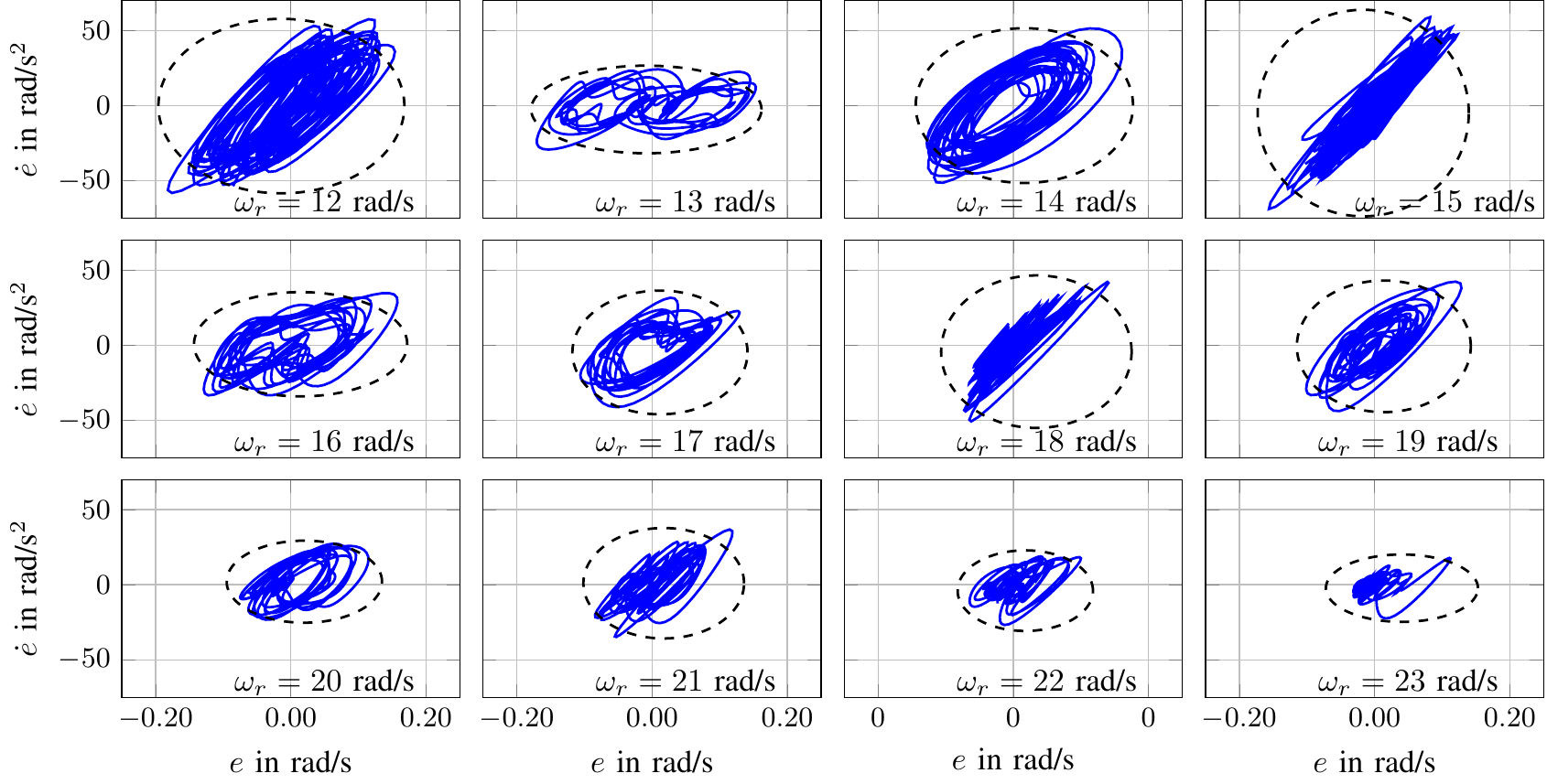}
		\caption{Phase plots for constant velocity reference $\omega_r$. The enclosing ellipse is drawn to highlight the decrease of the width of the limit cycles for increasing frequency of the perturbation.}\label{fig:limit_cycles_constant_speed}
	\end{center}
\end{figure*}

\begin{figure*}[htbp]
\begin{center}
	\includegraphics[width = 0.95\textwidth]{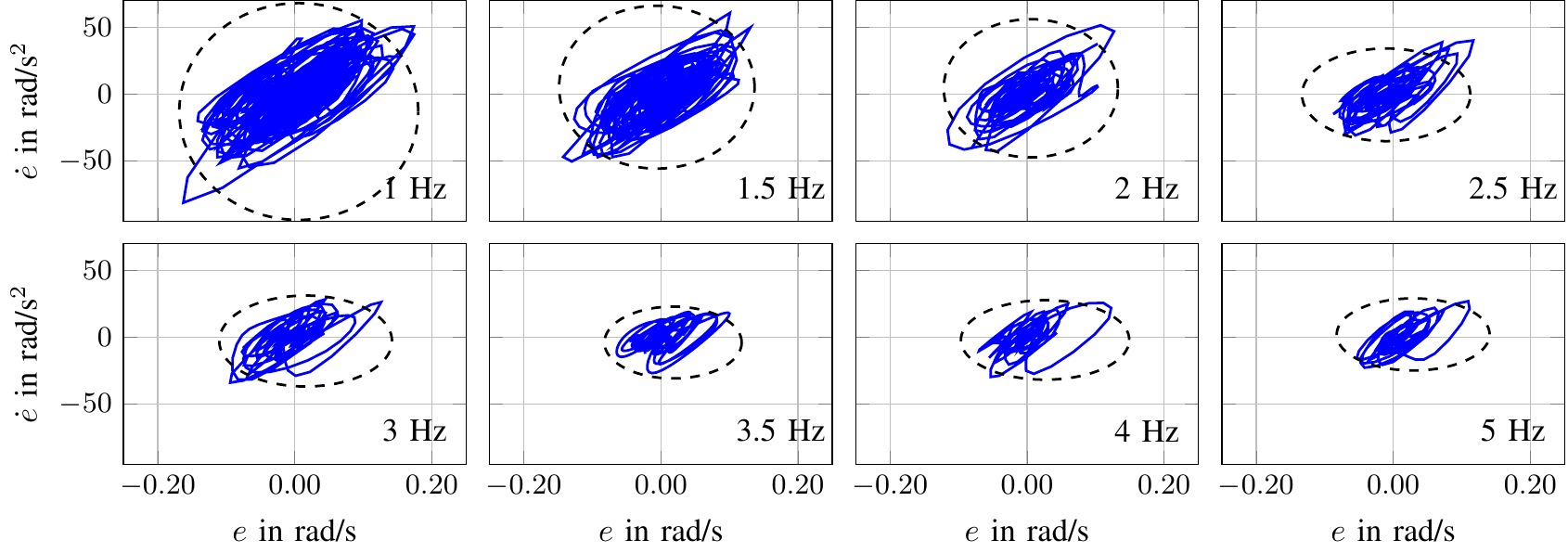}
	\caption{Phase plots for sinusoidal velocity reference $\omega_r(t) = \frac{100}{2\pi f_i}\cos(2\pi f_i t)$. The enclosing ellipse is drawn to highlight the decrease of the width of the limit cycles for increasing frequency of the perturbation.}\label{fig:limit_cycles_sinusoidal}
\end{center}
\end{figure*}

%% file: Tables/Constant_speed_scenarios.tex
\begin{tabular}{l|cccccc}\hline
	$\omega_r$ (rad/s)	& 12 & 13 & 14 & 15 & 16 & 17\T\B \\\hline
	$\max\vert e(t) \vert$	& 0.182  &  0.170  &  0.161  &  0.156  &  0.157  &  0.128\T\B \\
	$w_1^{max}$	& 0.565  &  0.482  &  0.415  &  0.362  &  0.318  &  0.282\T\B \\
	\specialrule{.2em}{.1em}{-1.2em}\\
	$\omega_r$ (rad/s)	& 18 & 19 & 20 & 21 & 22 & 23\T\B \\\hline
	$\max\vert e(t) \vert$	& 0.141  &  0.129 & 0.115  &  0.114 &  0.099  &  0.112\T\B \\
	$w_1^{max}$	& 0.251  &  0.226 & 0.204  &  0.185  &  0.168  &  0.154\T\B \\\hline
\end{tabular}

%% file: Tables/sinusoidal_speed_scenarios.tex
\begin{tabular}{l|cccccccc}\hline
	$f$ (Hz)	& 1 & 1.5 & 2 & 2.5 & 3 & 3.5 & 4 & 5\T\B \\\hline
	$\max\vert e(t) \vert$	& 0.174  &  0.143  &  0.127  &  0.117  &  0.126  &  0.123  &  0.119  &  0.11\T\B \\
	$w_1^{max}$	& 3.063  &  1.361  &  0.766  &  0.49  &  0.34  &  0.25  &  0.191  &  0.123\T\B \\
	\specialrule{.2em}{.1em}{-1.2em}\\\hline
\end{tabular}

%% file: sinusoidalSpeedRegime.tex
Eight experiments with different sinusoidal velocity references at frequencies $f_i\in\{ 1, 1.5,\dotsc, 4, 5 \}$ Hz were carried out for testing the \gls{STSMC} closed-loop behaviour during motion reversals. Choosing $\omega_r(t) = \frac{100}{2\pi f_i}\cos(2\pi f_i t)$ ensured constant acceleration peak for all $f_i$. In this case the friction rate assumed its highest magnitude $100\left( \frac{2 T_C \alpha}{\pi} + \beta \right)$ for zero velocity, where the contribution from the cogging torques was zero. For maximum velocity, the friction contribution was very small, whereas that of the cogging torques was the dominant one. In all experiments $L$ was no larger than 20 Nm/s. The finite-time convergence gains were calculated as $\bar{k}_1 = 10.5$ and $\bar{k}_2 = 22$, while the actual gains were selected as $k_1 = 0.9$ and $k_2 = 19.65$ according to \eqref{eq:k_2_estimate} for the same accuracy specification $\eta = 0.2$ rad/s. Again, the velocity error trajectory was periodic with same periodicity as the perturbation rate as seen in Figure \ref{fig:signals_sinusoidal}. The associated phase plots in Figure \ref{fig:limit_cycles_sinusoidal} illustrating the limit cycles showed decreasing error bound for increasing perturbation frequency $f$. This is more clearly seen in the top Figure \ref{fig:w_1_T} where a reciprocal quadratic relation is observed between $f$ and the error bound. Similarly to the constant speed experiments, Table \ref{tab:sinusoidal_speed_scenarios} shows the comparison between the real error bound and the one estimated based on \eqref{eq:w_1_bound_1} with $n = 0.5$.

%% file: conclusions.tex
Experimental validation of the stability properties of under-tuned super-twisting sliding mode control loops under the effect of periodic perturbations was pursued in this paper. Based on a series of tests performed on a commercial industrial motor, it was demonstrated that the existence of a limit cycle of the same period as the perturbation can be guaranteed under milder controller gain conditions compared to those required for finite-time stability. The experimental results verified that the width of the limit cycle quadratically decreases for smaller perturbation periods. Moreover, the controller gains provided by the proposed tuning guidelines were successfully applied to the real system ensuring that the considered accuracy specifications were satisfied. Future work will focus on reducing the conservatism in estimating the accuracy bounds in terms of both the perturbation profile and in relation to estimating the period fraction $n$.